\documentclass[11pt,reqno]{amsart}

\usepackage[T1]{fontenc}
\usepackage{lmodern}
\usepackage[margin=1in]{geometry}
\usepackage{amsmath,amssymb,mathtools}
\usepackage{microtype}
\usepackage{enumitem}
\usepackage{xcolor}
\usepackage[colorlinks=true,linkcolor=blue!45!black,
  citecolor=blue!45!black,urlcolor=blue!45!black]{hyperref}
\usepackage[nameinlink,capitalise,noabbrev]{cleveref}

\newtheorem{theorem}{Theorem}[section]
\newtheorem{lemma}[theorem]{Lemma}
\newtheorem{proposition}[theorem]{Proposition}

\theoremstyle{definition}
\newtheorem{definition}[theorem]{Definition}
\theoremstyle{remark}

\newcommand{\E}{\mathbb E}
\newcommand{\Prb}{\mathbb P}
\newcommand{\ind}{\mathbf 1}
\newcommand{\Dir}{\mathcal E}
\newcommand{\TV}{\mathrm{TV}}
\newcommand{\KL}{D_{\mathrm{KL}}}
\newcommand{\norm}[1]{\left\lVert #1\right\rVert}
\newcommand{\pos}[1]{\left(#1\right)_+}
\newcommand{\Cost}{\mathsf{Cost}}
\newcommand{\OPT}{\mathsf{OPT}}

\numberwithin{equation}{section}
\setlist[enumerate]{leftmargin=*,itemsep=3pt,topsep=5pt}
\allowdisplaybreaks[1]

\title{Transposition achieves OPT$+O(1)$ in polynomial time for IID list update}
\author{Clayton Mizgerd}
\address{Department of Mathematics, Statistics, and Computer Science, University of Illinois Chicago, Chicago, IL 60607, USA}
\email{cmizge2@uic.edu}
\date{}

\begin{document}
\begin{abstract}
In the classical list update problem, a set of items must be stored in a list-type structure, where accessing the $i$-th element costs $i$.  Items will be queried in an IID manner according to some probability distribution $p$ on the items.  We want to minimize the expected cost of each query.  The optimal order is to place the items in decreasing order of probability $p_1 \geq p_2 \geq \cdots$ with expected cost $\mathsf{OPT} = \sum_j j p_j$, but the probability vector $p$ is generally unknown.  Thus we use a self-organizing list following the transposition rule: an item is transposed 1 position forward whenever it is queried.  Coester (2026) proved that, at stationarity measure for the transposition rule, the expected cost of a query is at most $\mathsf{OPT} + 1$.  However, this Markov chain may have arbitrarily slow mixing time.  We prove that, for arbitrary $p$ and arbitrary initial orderings $\sigma$, after polynomially many queries in the number of items, the expected cost of a query is at most $\mathsf{OPT} + O(1)$.
\end{abstract}
\maketitle

\section{Introduction}

The list update problem is a classical and well-studied problem in theoretical computer science.  The user has a set of $n$ items labeled $1,\ldots,n$ that they wish to store in a linked list.  At each time step, a query arrives for a specific item $i$, which will need to be fetched from the list.  Fetching an item in position $j$ costs $j$.  In this work, we will focus on the independent and identically distributed (IID) case: there is a fixed probability vector $p = (p_1,\ldots,p_n)$ on the items and each query samples item $I$ according to $\mathbb{P}[I=i] = p_i$.  Relabel the items so that
\[ 1 \geq p_1 \geq p_2 \geq \cdots \geq p_n \geq 0, \qquad \sum_{i=1}^n p_i = 1, \]
breaking ties arbitrarily.  Thus if the list order is $\sigma$, then the expected cost of a query is
\[ \Cost(\sigma) = \sum_{i=1}^n p_i \sigma(i). \]

It is not hard to see that the optimal ordering is the identity ordering $\sigma(i) = i$, ordering by decreasing probability $p_i$.  Thus let
\[ \mathsf{OPT} := \Cost(\operatorname{id}) = \sum_{i=1}^n i p_i. \]

While the IID list update model is reasonable, in applications, the user is generally unaware of the probability vector $p$ and so cannot access this order.
One natural solution is to record queries for an empirical estimate of $p$, but this adds overhead to the list structure which is often undesirable.  Thus another commonly-studied strategy is a dynamic rule.

This study dates back to the work of McCabe~\cite{mccabe1965serial}.  There are two common dynamic rules for self-organizing lists.  The first is called move-to-front, where when an item $i$ is queried, it is moved to the front of the list, preserving the relative order of the remaining items.

The second rule, which we will be studying, is the transposition rule.  Under the transposition rule, when an item $i$ is queried, it is moved one position forward (transposed with its front neighbor).  If the queried item is already in the front, then the list is not changed.

Rivest \cite{rivest1976selforganizing} established the stationary distribution, and proved that the expected cost at stationarity for the transposition rule is no larger than that of move-to-front for any probability vector $p$.  Rivest also conjectured that the transposition rule is optimal among rules depending only on the position of the queried item.  This was disproven by Anderson, Nash, and Weber~\cite{anderson1982counterexample} with a six-item example.  However, near-optimality was shown by a line of work~\cite{makjamroen1991self,gamarnik2005transposition,indyk2025optimal} culminating in the following result of Coester~\cite{coester2026transposition}.

\begin{theorem}[{\cite[Theorem 1]{coester2026transposition}}]
    Let $\mu$ denote the stationary measure for the transposition rule.  Then
    \[ \mathbb{E}_{\sigma \sim \mu}[ \Cost(\sigma) ] \leq \OPT + 1. \]
\end{theorem}

While these guarantees hold at the stationary measure $\mu$, the mixing time cannot be bounded in terms of $n$ uniformly over all $p$.  Indeed, one may choose $p_i \propto 1/f(i)$ for arbitrarily fast-growing functions $f$ and make the spectral gap of the transition kernel $P$ arbitrarily small.
Under pathological distributions with very small tails, the failure of the tail elements of the list to mix is immaterial to controlling the cost.  Coester~\cite{coester2026transposition} conjectured that the transposition rule achieves expected cost $\OPT + O(1)$ after polynomially many requests under any $p$ and any initial ordering.  Our main result is to confirm this conjecture.

\begin{theorem}\label{thm:main}
    For every $n \geq 1$, every probability vector $p$ on $[n]$, every initial permutation $\sigma_0$, and every integer $t \geq n^{29}$,
    \[ \mathbb{E}[ \Cost(\sigma_t) ] \leq \OPT + 1250, \]
    where $\sigma_t$ is the permutation after $t$ steps.
\end{theorem}

Our proof shows $C=1250$ suffices.  We have not made any effort to optimize the exponent $29$ or the constant $1250$.

\subsection{Proof overview}

Define the excess cost
\begin{equation} \label{eq:cost}
F(\sigma) = \Cost(\sigma) - \OPT = \sum_{i < j} (p_i - p_j) \mathbf1\{ \sigma(i) > \sigma(j) \}. \end{equation}

To verify the identity, notice that one can sort the list by adjacent transpositions, and each reduces the search cost by $p_i - p_j$ and removes exactly that inversion $(i,j)$.  It will be convenient to work on a logarithmic scale; let $y_i = \log p_i$.

For $x \in \mathbb{R}$, let $I_x(\sigma)$ count inverted pairs
whose logarithmic probabilities lie on opposite sides of $x$:
\[
 I_x(\sigma)=\#\{(i,j) : y_j\le x<y_i,\
                         \sigma(i)>\sigma(j)\}.
\]
Since all $y_i\le0$, we have
\begin{align*}
 F(\sigma)
 &=\sum_{i<j}\left(\int_{y_j}^{y_i}e^x\,dx\right)
                  \ind\{\sigma(i)>\sigma(j)\} =\int_{-\infty}^0e^xI_x(\sigma)\,dx.
\end{align*}

Fix some $x < 0$; we will control $I_x(\sigma)$.  Define an adaptive buffer $\theta = \theta(x)$ so that
\begin{equation} \label{eq:defining-sum}
\sum_{i=1}^n \pos{2\theta-|x-y_i|}=1.
\end{equation}
For technical reasons, we cap $\theta$ at $1/2$.
Observe that 
\begin{equation}
 \#\{i:|y_i-x|\le\theta(x)\}\le\frac1{\theta(x)};
 \label{eq:scale-local}
\end{equation}
indeed, each such $i$ contributes $\theta(x)$ to the left hand side of \eqref{eq:defining-sum}.  (The cap on $\theta$ can only help.)

Call an inversion $(i,j)$ unsurprising if $x-\theta < y_i,y_j < x+\theta$.  The number of unsurprising inversions around $x$ is deterministically bounded by $\theta^{-2}$ by \eqref{eq:scale-local}.  We show that $\int e^x \theta(x)^{-2} dx = O(1)$, so unsurprising inversions have $O(1)$ total contribution.

For each surprising inversion $(i,j)$ we have $y_i - y_j \geq \theta$, and so $p_i/p_j \geq e^\theta$.  Thus the bias is $(p_i - p_j)/(p_i + p_j) \geq \tanh(\theta/2)$, and the number of inversions can be controlled by comparison with the number of inversions in an exclusion process with bias $\tanh(\theta/2)$.  This gives $O(\theta^{-2})$ expected surprising inversions at stationarity, another $O(1)$ contribution to $\mathbb{E}[F]$.

However, while the exclusion chain converges to stationarity in polynomial time, the original chain will not for small $x$.  We can bound the contribution to $F$ of the distance from $\nu$ to stationarity $\mu$ by
\[ n^{O(1)} \sqrt{\Dir(f,\log f)}, \qquad f = \frac{d\nu}{d\mu}. \]

This Dirichlet form is the continuous-time rate of entropy decay, related to the local stationarity notion of Liu, Mohanty, Raghavendra, Rajaraman, and Wu~\cite{liu2024locally}.  We show a discrete-time analog of their results to bring this error to $0$ even if the chain has not mixed.

\subsection*{Acknowledgments}

We thank Christian Coester for introducing us to this problem.  We thank Vishesh Jain and June Vuong for helpful conversations.  The author is supported by a Simons Dissertation Fellowship.

\subsection*{AI usage statement}

The author first worked on this problem in April 2026 without AI, and identified $p_j \propto \exp(-j/\sqrt n)$ as a difficult example.  GPT-6 Astra was able to prove the result for this specific vector.  The author  simplified the proof, and then prompted Astra to generalize the new proof to arbitrary $p$.  This proof was then rewritten.  Codex was used to assist in preparing the manuscript.  All mathematical content and errors are the sole responsibility of the author.

\section{The adaptive buffer}

\subsection{Preliminaries}

Until \cref{sec:entropy}, we will assume $p_i > 0$ for all $i$.  Set $y_i = \log p_i$.

Let $P = P_p$ be the described kernel for the transposition rule.  The kernel $P$ is reversible with respect to
\begin{equation}\label{eq:def-of-mu}
 \mu(\sigma)=\frac1Z\prod_{i=1}^n p_i^{\,n-\sigma(i)}.
\end{equation}
Indeed, moving $j$ ahead of its predecessor $i$ has forward
probability $p_j$, reverse probability $p_i$, and stationary
weight ratio $p_j/p_i$.

We define total variation distance
\( \norm{\alpha-\beta}_{\TV}=\frac12\sum_z|\alpha(z)-\beta(z)|. \)

For any $\mu$-reversible transition matrix $B$, write
\[
 \Dir_B(u,v)=\frac12\sum_{\sigma,\eta}\mu(\sigma)B(\sigma,\eta)
 (u(\sigma)-u(\eta))(v(\sigma)-v(\eta))
 =\mu[u(I-B)v].
\]

We use the convention $0\log0=0$.  For a nonnegative density $f$,
interpret $\Dir_B(f,\log f)$ by its nonnegative edge sum.
Edges of zero conductance contribute $0$, as do edges on which
both values of $f$ vanish.  A positive-conductance edge with
exactly one zero value of $f$ contributes $+\infty$.

\subsection{Integrability of the buffer}

For each $x\in\mathbb R$, let $r(x)>0$ be the unique solution of
\begin{equation} \label{eq:defining-sum-r}
 \sum_{i=1}^n \pos{r(x)-|x-y_i|}=1 \qquad \text{and} \qquad \theta(x) = \frac12 \min\{1,r(x) \}.
\end{equation}
Here $a_+=\max\{a,0\}$.  The left hand side, as a function of $r$,
is continuous, is initially zero, and is strictly increasing once
positive, with limit $+\infty$.  Thus the solution exists and is unique.

When many logarithmic probabilities are near $x$, the scale $\theta(x)$
is small.  The following lemma makes precise the two estimates we
need from this choice.

\begin{lemma}
\label{lem:scale}
The function $\theta$ is $1/2$-Lipschitz and
$1/(2n)\le\theta(x)\le1/2$.  Furthermore,
\begin{equation}
 \int_{-\infty}^0\frac{e^x}{\theta(x)^2}\,dx\le 4+32e.
 \label{eq:scale-integral}
\end{equation}
\end{lemma}
\begin{proof}
The left hand side of \eqref{eq:defining-sum-r} is at most $n \cdot r(x)$, so $r(x)\ge1/n$ and $\theta(x) \geq 1/(2n)$.
For $x,z\in\mathbb R$, the triangle inequality gives
\[
 \sum_{i=1}^n \pos{r(x)+|z-x|-|z-y_i|}
 \ge\sum_{i=1}^n \pos{r(x)-|x-y_i|}=1.
\]
Consequently $r(z)\le r(x)+|z-x|$.  Thus $r$ is $1$-Lipschitz and so $\theta$ is $1/2$-Lipschitz.
Notice that
\begin{equation*}
 \frac{1}{\theta(x)^2} \leq 4 + \frac{1}{\theta(x)^2} \sum_{i=1}^n \pos{2\theta(x)-|x-y_i|}.
\end{equation*}
If $\theta(x) = 1/2$, then clearly $1/\theta(x)^2 = 4$ and the sum is nonnegative.  If $\theta(x) \ne 1/2$, then the sum is exactly $1$.  Thus
\begin{align*}
    \int_{-\infty}^0 \frac{e^x}{\theta(x)^2} dx & \leq \int_{-\infty}^0 e^x \left( 4 + \frac{1}{\theta(x)^2} \sum_{i=1}^n (2\theta(x) - |x-y_i|)_+ \right) dx \\
    & = 4 + \sum_{i=1}^n \int_{-\infty}^0 e^x \frac{(2\theta(x) - |x-y_i|)_+}{\theta(x)^2} dx.
\end{align*}

Observe that if the numerator is nonzero, then $|x - y_i| \leq 2\theta \leq 1$.  Thus $e^x \leq e^{y_i+1} = e p_i$, giving us the bound
\begin{align*}
    \int_{-\infty}^0 \frac{e^x}{\theta(x)^2} dx \leq 4 + e \sum_{i=1}^n p_i \int_{-\infty}^0 \frac{(2\theta(x) - |x-y_i|)_+}{\theta(x)^2} dx.
\end{align*}

Fix $i$; we now bound this integral by $32$.  By Lipschitz continuity, we may bound the numerator by $2\theta(y_i)$.  We must lower bound the denominator in two cases.

Write $\delta=|x-y_i|$.  If $\delta\le\theta(y_i)$, Lipschitz continuity gives $\theta(x) \geq \theta(y_i)/2$.  If $\delta > \theta(y_i)$, then positivity of the numerator implies $2\theta(x) > \delta$.  Thus we have
\[ \int_{-\infty}^0 \frac{(2\theta(x) - |x-y_i|)_+}{\theta(x)^2} dx \leq 2\theta(y_i) \cdot \frac{2\theta(y_i)}{\theta(y_i)^2/4} + 2 \int_{\theta(y_i)}^\infty \frac{2\theta(y_i)}{\delta^2/4} d\delta \leq 32. \]

Thus
\[ \int_{-\infty}^0 \frac{e^x}{\theta(x)^2} dx \leq 4 + 32e. \qedhere \]
\end{proof}

\subsection{Discarding local contributions}

For real threshold $a$ and buffer $b>0$, define
\[
 J_{a,b}(\sigma)=\#\{(i,j):y_j\le a,\ y_i>a+b,\
                              \sigma(j)<\sigma(i)\}.
\]
The inequalities force $p_i>p_j$, so the weak ordering of the labels
automatically gives $i<j$.  The items need not be adjacent.
For each $x<0$, put
\[
 J_x=J_{x,\theta(x)}.
\]

\begin{lemma}
\label{lem:separated-reduction}
For every permutation $\sigma$,
\[
 F(\sigma)\le2+(1+\sqrt e)\int_{-\infty}^0
       e^xJ_x(\sigma)\,dx.
\]
\end{lemma}
\begin{proof}
Since $\theta$ is $1/2$-Lipschitz and bounded, the map
$x\mapsto x+\theta(x)$ is strictly increasing and maps
$\mathbb R$ onto $\mathbb R$.  For each label $i$, let $b_i$ be
the unique solution of $b_i+\theta(b_i)=y_i$.

First consider inverted pairs $(i,j)$ with $y_j\ge b_i-\theta(b_i)$.
Since $y_j\le y_i=b_i+\theta(b_i)$, all such lower endpoints lie in
$[b_i-\theta(b_i),b_i+\theta(b_i)]$.  By \eqref{eq:scale-local}, there
are at most $1/\theta(b_i)$ of them.  Each contributes at most
\[
 p_i-p_j\le p_i(1-e^{-2\theta(b_i)})\le2\theta(b_i)p_i
\]
to $F$.  Summing first over these lower endpoints and then over
$i$ bounds their total contribution by $2\sum_i p_i=2$.

For any remaining inverted pair, $y_j<b_i-\theta(b_i)$.  Monotonicity of
$x+\theta(x)$ shows that this pair is counted by $J_x$ exactly
when $y_j\le x<b_i$.  Since $b_i<y_i\le0$, its contribution
to the integral in the statement is $e^{b_i}-e^{y_j}$.  Moreover,
\[
 \frac{p_i-p_j}{e^{b_i}-e^{y_j}}
 =e^{\theta(b_i)}+\frac{e^{\theta(b_i)}-1}{e^{b_i-y_j}-1}
 \le e^{\theta(b_i)}+1\le1+\sqrt e,
\]
where we used $\theta(b_i)\le1/2$.
Sum over these remaining inversions and use \eqref{eq:cost}.
The other contributions to the integral are nonnegative, which
proves the claim.
\end{proof}

\section{Auxiliary exclusion chains}
\label{sec:auxiliary}

Throughout this section, assume that $p_i>0$ for every label $i$.
We only need separation scales $0<\theta\le1/2$, since
\cref{lem:scale} places $\theta(x)$ in this range.
Fix two separated probability classes.  The auxiliary chain freezes
the intermediate labels and the relative order within each class,
while allowing adjacent low--high pairs in the remaining subsequence
to exchange.  We first bound its binary inversion count, then realize
its exchanges by short sorting paths in the original chain.

\subsection{The biased exclusion chain}

We first control a homogeneous chain which sorts elements of relative weights $1,e^{-\theta}$.

\begin{definition}[Biased exclusion chain]
\label{def:exclusion}
Fix $N\ge2$ and $0<\theta\le1/2$.
The biased exclusion chain
$(Y_u)$ has as its states binary words of length $N$ with fixed
positive numbers of $L$'s and $H$'s.  At each step choose one of
the $N-1$ adjacent pairs uniformly.  When its letters differ, resample their order,
putting $H$ first with probability $(1+e^{-\theta})^{-1}$
and $L$ first otherwise.  If they agree, leave the word unchanged.
Let $\zeta$ be the stationary law of the chain.
\end{definition}

\begin{lemma}
\label{lem:exclusion}
Let $V(Y)$ denote the number of pairs in which an $L$ precedes an $H$.  Then for any $n \geq 1024$, $N \leq n$, and $0 < \theta \leq 1/2$, from an arbitrary start $Y_0$, the chain in \cref{def:exclusion} satisfies
\[
\E V(Y_{n^4})\le\frac14+\frac{\pi^2}{6\theta^2}.
\]
\end{lemma}
\begin{proof}
If \(0<\theta\le1/(2n)\), the conclusion follows from \(V\le N^2/4\le n^2/4\le\pi^2/(6\theta^2)\). We therefore assume \(\theta>1/(2n)\) in the mixing argument.

Let $h$ be the number of high letters.
Detailed balance shows that $\zeta(Y) \propto e^{-\theta V(Y)}$.

Let $g_j$ be the number of low letters with exactly $j$ high letters succeeding them.  Notice $V = \sum_j j g_j$.  Suppose $g_j \geq r$.  Then we can move this block of $r$ letters $L$ to the end of the word and decrease $V$ by $jr$.  This gives a bijection $\{ g_j \geq r \}$ to $\{ g_0 \geq r \}$.  Thus
\[
 \Prb_\zeta\{g_j\ge r\}
 =e^{-\theta jr}\Prb_\zeta\{g_0\ge r\}\le e^{-\theta jr}.
\]
Summing these tail bounds yields
\[ \E_\zeta V
 \le\sum_{j=1}^h j\sum_{r\ge1}e^{-\theta jr}
 =\sum_{j=1}^h\frac{j}{e^{\theta j}-1}
 \le\int_0^\infty\frac{j}{e^{\theta j}-1}\,dj
 =\frac{\pi^2}{6\theta^2}.
\]

Viewing the $L$'s as particles, our chain is exactly the discrete-time
exclusion process of Levin--Peres
\cite{levin2016mixing}, with bias
\[
 \beta=\tanh(\theta/2)
 \ge\frac\theta2-\frac{\theta^3}{24}
 \ge\frac{\sqrt2}{3}\theta.
\]
Here we used $\tanh u\ge u-u^3/3$ and $\theta\le1/2$.  Levin and Peres~\cite[Theorem~10]{levin2016mixing} prove the total variation mixing time from a worst-case start satisfies
\begin{align*}
 t_{\mathrm{mix}}(n^{-2})
 &\le
 \frac{9N}{\theta^2}
 \left(4\log n+\frac{\theta N}{2}\right) 
 \le n^4.
\end{align*}
Since $V \le N^2/4\le n^2/4$, the stationary estimate above now yields
\[
 \E V(Y_{n^4})
 \le \E_\zeta[V]
       +\frac{n^2}{4}\|\mathcal L(Y_{n^4})-\zeta\|_{\TV}
 \le \frac{\pi^2}{6\theta^2}+\frac14. \qedhere
\]
\end{proof}

\subsection{Comparison with the auxiliary chain}

We now define an auxiliary chain which uses the actual exchange probabilities.

\begin{definition}[Auxiliary chain]
\label{def:auxiliary}
Fix $x \in\mathbb R$ and $0 < \theta\le1/2$.  Let $B_{x,\theta}$ denote the following kernel.  Call label $i$ \emph{low} if $y_i \leq x$, \emph{intermediate} if $x < y_i \leq x + \theta$, and \emph{high} if $y_i > x + \theta$.
If fewer than two labels are non-intermediate, set $B_{x,\theta} = I$.

Choose uniformly a pair $(i,j)$ whose labels are adjacent in the subsequence
of non-intermediate labels.
If $i$ and $j$ are both high or both low, reject the move.
Otherwise, let $d = \sigma(j) - \sigma(i)$ denote the physical distance
(i.e.~the number of intermediate labels between them plus $1$).
Switch $i$ and $j$ with probability
\begin{equation}\label{eq:heat-bath} \frac{p_j^d}{p_j^d + p_i^d}. \end{equation}
\end{definition}


Observe that \eqref{eq:heat-bath} is the heat-bath rule with respect to $\mu$ as defined in \eqref{eq:def-of-mu}.  Thus $\mu$ is a stationary measure; however, this is generally not the unique stationary measure.

\begin{proposition}
\label{prop:auxiliary}
For every $n\ge1024$, $x\in\mathbb R$, $0 < \theta\le1/2$, and permutation $\sigma$,
\begin{equation}
 \bigl(B_{x,\theta}^{n^4}J_{x,\theta}\bigr)(\sigma)
 \le \frac14+\frac{\pi^2}{6\theta^2}.
 \label{eq:aux-relaxation}
\end{equation}
For every real-valued function $v$ on permutations,
\begin{equation}
 \Dir_{B_{x,\theta}}(v,v)\le2n^2e^{-x}\Dir_P(v,v).
 \label{eq:aux-comparison}
\end{equation}
\end{proposition}

The factor $e^{-x}$ reflects the request probabilities used by the
comparison paths: every requested label has probability at least $e^x$.

\begin{proof}
If either class is empty, $J_{x,\theta} \equiv 0$, $B_{x,\theta} = I$, and so both conclusions are immediate.

Write $B=B_{x,\theta}$.  Define the word $W_0$ by deleting the intermediate labels from
$\sigma$ and replacing the remaining labels by $H$ or $L$ to denote high or low.
Let $N$ be the length of this word and $h$ its number of high letters.

We first prove \eqref{eq:aux-relaxation} via coupling with the exclusion process analyzed above.  Let $W_t$
be the binary word of the auxiliary chain and let $Y_t$
be the homogeneous chain of \cref{def:exclusion}, both started
from the word $W_0$ constructed.
At each step, perfectly couple the adjacent pair chosen and maximally couple the Bernoulli random variable determining whether or not to put $H$ in front.  Notice that, for a high label $u$ and a low label $\ell$ at distance $d$, $W_t$ puts the high label $H$ in front with probability
\[
 \frac{p_u^d}{p_u^d+p_\ell^d}
 =\frac{1}{1+e^{-d(y_u-y_\ell)}}
 \ge\frac{1}{1+e^{-\theta}}.
\]

We claim that this coupling keeps $W_t$ at least as well sorted
as $Y_t$.  For a word $\eta$, let $s_k(\eta)$ count its high
letters in the first $k$ positions.  An update at edge $k$
changes only $s_k$, replacing it by
\[
 s_k'=
 \begin{cases}
 \min\{s_{k-1}+1,s_{k+1}\},&\text{for the $H$-first outcome},\\
 \max\{s_{k-1},s_{k+1}-1\},&\text{for the $L$-first outcome}.
 \end{cases}
\]
Both expressions are increasing in the neighboring prefix
counts, and the first is at least the second.  For a same-class
pair they coincide.  Thus a shared outcome preserves prefix
order, and the larger sorting probability in $W_t$ can only
increase its updated count.  Since $W_0=Y_0$, induction gives
$s_k(W_t)\ge s_k(Y_t)$ for every $k$ and $t$.

It follows that the $j$th high letter of $W_t$ is no farther
right than the $j$th high letter of $Y_t$.  If the high letters
of a word occupy positions $c_1<\cdots<c_h$, its inversion
count is $V=\sum_{j=1}^h(c_j-j)$.  Hence $V(W_t)\le V(Y_t)$,
and \cref{lem:exclusion} gives
\[
  (B^{n^4}J_{x,\theta})(\sigma)
  =\E V(W_{n^4})
  \le\E V(Y_{n^4})
 \le\frac14+\frac{\pi^2}{6\theta^2}.
\]

We now prove \eqref{eq:aux-comparison} via a canonical paths argument.  Let $\mathsf E_B$ and $\mathsf E_P$ be the sets of undirected
edges between distinct states with positive $B$- and $P$-transition
probabilities, respectively.  Given an exchange in $\mathsf E_B$, orient it so that it travels from its endpoint of smaller stationary weight to higher:
\[
 (\ell,g_1,\ldots,g_m,u)
 \longrightarrow
 (u,g_1,\ldots,g_m,\ell),
\]
where $\ell$ is low, $u$ is high, and all $g_j$ are intermediate labels.  We construct the path in $\mathsf E_P$ by $m+1$ consecutive calls to $u$ followed by calling $g_1,g_2,\ldots,g_m$ in order to put $\ell$ at the end.  Every requested label is at least intermediate and so has $p_i \geq e^x$.  Each move increases the weight, so each intermediate state has measure at least $\mu(\sigma)$.  Thus each original edge has conductance at least $e^x \mu(\sigma)$.  Each edge in the auxiliary chain has conductance at most $\mu(\sigma)$.

Each path has length at most $2m+1 \leq 2n$.
We claim that any edge in $\mathsf E_P$ can be used by at most $n$ paths in $\mathsf E_B$, as the original position $w$ of the high-label vertex is sufficient to learn the transition in $\mathsf E_B$.  If the called label is high, then we are in the process of switching this label with its nearest left non-intermediate neighbor, which will then be pushed back to index $w$.  If the called label is intermediate, then it is switching with a low label, which will fall back to index $w$, and the nearest left non-intermediate neighbor of the low vertex is the high vertex with which it is switching.

Thus by a standard canonical paths argument (see, e.g.~\cite{diaconis1993comparison})
\begin{align*}
 \Dir_B(v,v)
 &\le\sum_{\{\eta,\xi\}\in\mathsf E_P}
       \bigl(v(\xi)-v(\eta)\bigr)^2
       \sum_{\substack{e=\{\sigma,\tau\}\in\mathsf E_B\\
                       \{\eta,\xi\}\in\gamma_e}}
          |\gamma_e|\mu(\sigma)B(\sigma,\tau)\\
 &\le(2n)(n)e^{-x}
       \sum_{\{\eta,\xi\}\in\mathsf E_P}
          \mu(\eta)P(\eta,\xi)\bigl(v(\xi)-v(\eta)\bigr)^2\\
 &=2n^2e^{-x}\Dir_P(v,v).
\end{align*}
This proves \eqref{eq:aux-comparison}.
\end{proof}

We will apply this proposition to the kernel $B_{x,\theta(x)}$.

\section{A cost bound via entropy decay}

\subsection{From cost to Dirichlet forms}

We now handle $\nu(F)$ in terms of a Dirichlet form.  Set
\begin{equation}
 C_0=2+(1+\sqrt e)\left[\frac14+\frac{2\pi^2}{3}(1+8e)\right].
 \label{eq:absolute-constant}
\end{equation}
Only the fact that $C_0$ is an absolute constant will matter.

\begin{lemma}
\label{prop:quadratic}
For any strictly positive probability vector, $n\ge1024$, and any
probability measure $\nu$ on permutations with density $f=d\nu/d\mu$,
\[
 \nu(F)\le C_0+\frac32n^5\sqrt{\Dir_P(f,\log f)}.
\]
\end{lemma}
\begin{proof}
Averaging \cref{lem:separated-reduction} over $\nu$ gives
\begin{align*}
    \nu(F) & \leq 2 + (1+\sqrt e) \int_{-\infty}^0 e^x \nu(J_x) dx & \text{\cref{lem:separated-reduction}} \\
    & = 2 + (1+\sqrt e) \left[ \int_{-\infty}^0 e^x \nu((1 - K_x) J_x) dx + \int_{-\infty}^0 e^x \nu(K_x J_x) dx \right] & K_x := B_x^{n^4} \\
    & \leq 2 + (1+\sqrt e) \left[ \int_{-\infty}^0 e^x \nu((1 - K_x) J_x) dx + \int_{-\infty}^0 e^x \left(\frac14+\frac{\pi^2}{6\theta^2}\right) dx \right] & \text{\cref{prop:auxiliary} \eqref{eq:aux-relaxation}} \\
    & \leq 2 + (1+\sqrt e) \left[ \int_{-\infty}^0 e^x \nu((1 - K_x) J_x) dx + \frac14+\frac{2\pi^2}{3}(1+8e) \right] & \text{\cref{lem:scale}} \\
    & = C_0 + (1+\sqrt e) \int_{-\infty}^0 e^x \nu((1-K_x) J_x) dx.
\end{align*}

Fix $x$ and write $J=J_x$ and $K=K_x$.  This kernel is
$\mu$-reversible.  Since $d\nu=f\,d\mu$,
\[
 \nu((1-K)J)=\Dir_K(f,J).
\]
The low and high classes have respective sizes $\ell_x,h_x$ with
$\ell_x+h_x\le n$, so $0\le J\le\ell_xh_x\le n^2/4$.
Factoring the difference of squares and applying Cauchy--Schwarz gives
\begin{align*}
 \nu((1-K)J)
 &=\Dir_K(f,J) = \frac12 \sum_{z,z'} \mu(z) K(z,z') (f(z) - f(z')) (J(z) - J(z')) \\
 &\le\sqrt{\Dir_K(\sqrt f,\sqrt f)} \left[\frac12\sum_{z,z'}\mu(z)K(z,z')
       \bigl(\sqrt{f(z)}+\sqrt{f(z')}\bigr)^2
       (J(z)-J(z'))^2\right]^{1/2}\\
& \leq \sqrt{\Dir_K(\sqrt f,\sqrt f)} \left[ \frac12 \sum_{z,z'} \mu(z) K(z,z') 2(f(z) + f(z')) \left(\frac{n^2}{4}\right)^2 \right]^{1/2} \\
  &\le\frac{n^2}{2\sqrt2}\sqrt{\Dir_K(\sqrt f,\sqrt f)}.
\end{align*}

The final inequality follows from reversibility and stationarity of $K$, together with $\mu(f) = \nu(1) = 1$.

Let $B = B_x$.
Using Cauchy--Schwarz and \cref{prop:auxiliary} \eqref{eq:aux-comparison},
\begin{multline*}
    \Dir_K(\sqrt f, \sqrt f) = \sum_{k=0}^{n^4-1} \Dir_B(\sqrt f, B^k \sqrt f) \leq \sum_{k=0}^{n^4-1} \left( \Dir_B(\sqrt f, \sqrt f) \right)^{1/2} \left( \Dir_B( B^k \sqrt f, B^k \sqrt f ) \right)^{1/2} \\
    \leq n^4\Dir_B(\sqrt f, \sqrt f)
    \leq 2n^6e^{-x}\Dir_P(\sqrt f,\sqrt f).
\end{multline*}
Keeping the dependence on $x$ inside the comparison integral gives
\begin{align*}
 (1+\sqrt e)\int_{-\infty}^0e^x
       \nu((1-K_x)J_x)\,dx& \le\frac{1+\sqrt e}{2\sqrt2}n^2
          \sqrt{2n^6\Dir_P(\sqrt f,\sqrt f)}
          \int_{-\infty}^0e^{x/2}\,dx\\
 &\le\frac32n^5\sqrt{4\Dir_P(\sqrt f,\sqrt f)}.
\end{align*}
Here we used $(1+\sqrt e)/2\le3/2$.
Finally, Cauchy--Schwarz shows that $4\Dir_P(\sqrt f, \sqrt f) \leq \Dir_P(f,\log f)$ for any choice of $P,f$, completing the proof of the claim.
\end{proof}

\subsection{Entropy decay}

We begin by recording a very general argument on entropic decay that holds for arbitrary reversible Markov chains on finite state spaces.  This may be thought of as a discrete-time analogue of \cite[Lemma III.1]{liu2024locally}.

For measures $\alpha \ll \beta$ on the same finite space, let
\[
 \KL(\alpha\Vert\beta)
 =\sum_z\alpha(z)\log\frac{\alpha(z)}{\beta(z)}.
\]

\begin{lemma}
\label{lem:entropy}
Let $P$ be a reversible transition matrix on a finite space with
a strictly positive stationary measure $\mu$, and put
$H=\log(1/\mu_{\min})$.  For any initial law $\nu_0$, let
\[
 \nu_s=\nu_0P^s,\qquad f_s=\frac{d\nu_s}{d\mu},\qquad
 g_t=\frac{1+f_t+f_{t+1}}3.
\]
For every integer $t\ge0$,
\begin{equation}
 \Dir_P(g_t,\log g_t)
 \le\frac{H+\log3}{3}\sqrt{\frac{H}{2(t+1)}}.
 \label{eq:entropy-smoothing}
\end{equation}
\end{lemma}
\begin{proof}
Write $H_s=\KL(\nu_s\Vert\mu)$.  Reversibility gives
$f_{s+1}=Pf_s$.  The log-sum inequality shows that relative entropy
decreases under application of a transition matrix.
In particular $0\le H_s\le H_0\le H$.
Consider the joint laws
\[
 A_s(x,y)=\nu_s(x)P(x,y),\qquad
 B_s(x,y)=\nu_{s+1}(y)P(y,x).
\]
If $A_s(x,y)>0$, then $\nu_{s+1}(y)>0$ and, by reversibility,
$P(y,x)>0$, so $B_s(x,y)>0$ as well.
Their likelihood ratio, wherever $A_s(x,y)>0$, is
$f_s(x)/f_{s+1}(y)$.  Summing the logarithm of this ratio gives
\begin{equation}
 \KL(A_s\Vert B_s)=H_s-H_{s+1}.
 \label{eq:edge-entropy}
\end{equation}
Recall Pinsker's inequality
$\norm{\alpha-\beta}_{\TV}^2\le\KL(\alpha\Vert\beta)/2$.
The first marginals of $A_s,B_s$ are $\nu_s,\nu_{s+2}$.
Pinsker's inequality, \eqref{eq:edge-entropy}, and contraction
under taking a marginal give
\[
 \norm{\nu_s-\nu_{s+2}}_{\TV}^2
 \le\frac{H_s-H_{s+1}}2.
\]
The distance on the left is nonincreasing in $s$, by contraction
under $P$.  Summing from $s=0$ through $t$ and using $H_0\le H$ gives
\[
 \norm{\nu_t-\nu_{t+2}}_{\TV}
 \le\sqrt{\frac{H}{2(t+1)}}.
\]

For the Dirichlet form, note that $\mu(g_t)=1$ and
\[
 \frac13\le g_t\le e^H,\qquad
 (I-P)g_t=\frac{f_t-f_{t+2}}3.
\]
Thus the values of $\log g_t$ lie in an interval of length at most
$H+\log3$.  The signed measure $\nu_t-\nu_{t+2}$ has total mass
zero, and its positive and negative parts each have mass
$\norm{\nu_t-\nu_{t+2}}_{\TV}$.
By reversibility,
\begin{align*}
 \Dir_P(g_t,\log g_t)
 &=\mu[(I-P)g_t\,\log g_t]
 =\frac13\mu[(f_t-f_{t+2})\log g_t]\\
 &\le\frac{1}{3}\norm{\nu_t-\nu_{t+2}}_{\TV} (\max(\log g_t) - \min(\log g_t)) \\
 &\le\frac{H+\log3}{3}\sqrt{\frac{H}{2(t+1)}}.\qedhere
\end{align*}
\end{proof}

\subsection{Extremely small and zero probabilities}
\label{sec:entropy}

For the transposition chain with $n\ge1024$, $F\ge0$ and $f_t\le3g_t$.
Applying \cref{prop:quadratic} to the law with strictly positive
density $g_t$ and using \cref{lem:entropy} therefore gives
\begin{align}
 \nu_t(F)
 &\le3\mu(g_tF)
 \le3C_0+\frac92n^5\sqrt{\Dir_P(g_t,\log g_t)}\nonumber\\
 &\le3C_0+\frac{3\sqrt3}{2}n^5\sqrt{H+\log3}
       \left(\frac{H}{2(t+1)}\right)^{1/4}.
 \label{eq:cost-at-time}
\end{align}

The preceding arguments apply to every positive vector, but
$H = \log(1/\mu_{\min})$ need not be bounded in terms of $n$ alone (and indeed may not be finite, as we allow $p_i = 0$).
We perturb the request probabilities by exponentially small amounts
and couple the resulting chain to the original chain up to time
$T=n^{29}$.

\begin{lemma}
\label{lem:regularization}
Let $p_1\ge\cdots\ge p_n\ge0$ be a probability vector on $n\ge2$
labels.  Set
\[
 q_i=\frac{p_i+e^{-n}}{1+ne^{-n}}\quad(1\le i\le n).
\]
Then $q$ is a positive probability vector with the same weak ordering
as $p$.  Start the $p$-chain $(X_t)$ and the $q$-chain $(Y_t)$ from
the same permutation.  For every integer $t\ge0$,
\begin{equation}
 \E_p F_p(X_t)
 \le \E_qF_q(Y_t)+(t+1)n^2e^{-n}.
 \label{eq:regularization-transfer}
\end{equation}
If $\mu_q$ is the stationary measure of the $q$-chain and
$H_q=\log(1/\mu_{q,\min})$, then
\begin{equation}
 H_q\le\frac{n^3}{2}.
 \label{eq:regularized-entropy}
\end{equation}
\end{lemma}
\begin{proof}
The entries of $q$ are positive, sum to $1$, and preserve the weak
ordering of the labels.  Since
$q_i-q_j=(p_i-p_j)/(1+ne^{-n})$, the excess-cost identity
\eqref{eq:cost} gives
\[
 F_p(\sigma)=(1+ne^{-n})F_q(\sigma)
 \qquad\text{for every permutation }\sigma.
\]

Let $v$ be the uniform distribution on labels.  Then
\[
 q=(1-\alpha)p+\alpha v,\qquad
 \alpha=\frac{ne^{-n}}{1+ne^{-n}},
\]
so $\norm{p-q}_{\TV}\le\alpha\le ne^{-n}$.  Maximally couple the
requests independently at each step.  The lists agree until the
first request mismatch, whose probability by time $t$ is at most
$tne^{-n}$.  Since $0\le F_p\le n$, the coupling gives
\[
 \E_pF_p(X_t)
 \le\E_qF_p(Y_t)+tn^2e^{-n}
 =(1+ne^{-n})\E_qF_q(Y_t)+tn^2e^{-n} \leq \E_q F_q(Y_t) + (t+1)n^2 e^{-n},
\]
proving \eqref{eq:regularization-transfer}.

Finally, the exponents in the stationary weight
\eqref{eq:def-of-mu} sum to $n(n-1)/2$.  Every entry of $q$ is at
least $e^{-n}/(1+ne^{-n})$, and there are $n!$ permutations, so
\begin{align*}
 H_q&\le\log(n!)+\frac{n(n-1)}2
       \bigl(n+\log(1+ne^{-n})\bigr)\le\frac{n^3}{2}. \qedhere
\end{align*}
\end{proof}

\begin{proof}[Proof of \cref{thm:main}]
We prove the theorem with $C=1250$.
If $n<1024$, use $F_p\le n<C$.
Assume $n\ge1024$ and set $T=n^{29}$.
Given $p$, define $q$ as in \cref{lem:regularization}, and start
the $q$-chain $(Y_t)$ from the same permutation as $(X_t)$.
We have
\begin{align*}
 \E_pF_p(X_T)
 &\le\E_qF_q(Y_T)+(T+1)n^2e^{-n}
 &&\text{by \eqref{eq:regularization-transfer}}\\
 &\le3C_0+\frac{3\sqrt3}{2}n^5\sqrt{H_q+\log3}
       \left(\frac{H_q}{2(T+1)}\right)^{1/4} +(T+1)n^2e^{-n}
 &&\text{by \eqref{eq:cost-at-time}}\\
 &\le3C_0+\frac{3\sqrt3}{2}n^5\sqrt{n^3/2+\log3}
       \left(\frac{n^3}{4(T+1)}\right)^{1/4} +(T+1)n^2e^{-n}
 &&\text{by \eqref{eq:regularized-entropy}}\\
 &\le3C_0+\frac{3\sqrt6}{4}+2n^{31}e^{-n}
 &&T=n^{29}\\
 &<3C_0+3<C.
\end{align*}

See \eqref{eq:absolute-constant} for the exact definition of $C_0$ to verify $C = 1250$ suffices.  For every integer $t\ge T$,
the Markov property therefore gives
\[
 \sup_\sigma P_p^tF_p(\sigma)
 \le\sup_\sigma P_p^TF_p(\sigma)<1250. \qedhere
\]
\end{proof}

\bibliographystyle{amsplain0}
\bibliography{main}

\end{document}